\documentclass[twocolumn]{article}
\usepackage[dvips]{hyperref}
\usepackage{amsmath}
\usepackage{amsthm}
\usepackage{amssymb}
\usepackage{multirow}
\usepackage{url}
\usepackage{amssymb}
\usepackage{graphicx}
\usepackage[round]{natbib}

\newtheorem{thm}{Theorem}

\title{A Regional Bayesian POT Model for Flood Frequency Analysis}
\date{}
\author{Mathieu Ribatet$^{1,2}$ \and Eric Sauquet$^2$ \and
  Jean-Michel Gr\'esillon$^2$ \and Taha B.M.J. Ouarda$^1$}

\begin{document}

\onecolumn{
\maketitle
\begin{center}
  $^1$INRS-ETE, University of Qu\'ebec, 490, de la Couronne Qu\'ebec,
  Qc, G1K 9A9, CANADA \newline
  $^2$Cemagref, 3 bis quai Chauveau CP 220, 69336 Lyon Cedex 09,
  FRANCE
\end{center}

\begin{abstract}
  \noindent
  Flood Frequency Analysis is usually based on the fitting of an
  extreme value distribution to the series of local
  streamflow. However, when the local data series is short, frequency
  analysis results become unreliable. Regional frequency analysis is a
  convenient way to reduce the estimation uncertainty. In this work,
  we propose a regional Bayesian model for short record length sites.
  This model is less restrictive than the Index Flood model while
  preserving the formalism of ``homogeneous regions''. Performance of
  the proposed model is assessed on a set of gauging stations in
  France. The accuracy of quantile estimates as a function of
  homogeneousness level of the pooling group is also analysed. Results
  indicate that the regional Bayesian model outperforms the Index
  Flood model and local estimators.  Furthermore, it seems that
  working with relatively large and homogeneous regions may lead to
  more accurate results than working with smaller and highly
  homogeneous regions.
  \\
  \noindent
  Key~words:~Regional Frequency Analysis -- Bayesian Inference -- Index
  Flood -- \textit{L-moments} -- Markov Chain Monte Carlo
\end{abstract}
}

\section{Introduction}
\label{sec:intro}

Flood frequency analysis is essential in preliminary studies to define
the design flood. Methods for estimating design flow usually consist
of fitting one of the distributions given by the extreme value theory
to a sample of flood events. If modelling exceedance over a threshold
is of interest, a theoretical justification
\citep{Fisher1928,Balkema1974,Pickands1975} exists for the use of the
Generalized Pareto distribution (GP).
\begin{equation}
  F(x) = 1 - \left( 1 +\frac{\xi\left(x - \mu\right)}{\sigma}\right)
  ^{-1/\xi}
  \label{eq:gpd}
\end{equation}
where $1+\xi\left(x-\mu\right)/\sigma>0$, $\sigma>0$. $\mu$, $\sigma$
and $\xi$ are the location, scale and shape parameters. This
distribution is defined for $\xi \neq 0$, and can be derived by
continuity in the case $\xi = 0$, corresponding to the Exponential
case:
\begin{equation}
  F(x) = 1 - \exp\left( - \frac{x - \mu}{\sigma} \right)
  \label{eq:expo}
\end{equation}
A comprehensive review of the Extreme Value Theory is given by
\citet{Embrechts1997} and \citet{Coles2001}.
\par
However, frequency analysis can lead to unreliable flood quantiles
when little data is available at the site of interest. A convenient
way to improve estimates of flood statistics is to incorporate data
from other gauged locations in the estimation procedures. This
approach is widely applied in hydrology and is known as Regional Flood
Frequency Analysis (\textbf{RFFA}). One of the most popular and simple
approaches privileged by engineers is the Index Flood method
\citep{Dalrymple1960}. The standard procedure allows: a) the
delineation of homogeneous regions, i.e.\ a set of sites which behave
- hydrologically and/or statistically - in the same way; b) the
derivation of a regional flood frequency distribution; c) and the
estimation of the parameters and quantiles at the site of interest.
\par
Regions are collections of gauged basins with similar site
characteristics related to the flood magnitude. The pooled stations
are not necessarily in the proximity of the site of interest. Forming
homogeneous regions can be achieved in various ways. Regions have
first been geographically established. More recent work promoted the
use of geographically non-contiguous regions
\citep{Burn1990,GREHYS1996}.  Recent research has defined the concept
of ``region of influence'' \citep{Acreman1989}.  Other techniques can
be used such as Artificial Neural Networks to identify groups of
stations \citep{Hall2002}.
\par
The Index Flood model assumes that flood distributions at all sites
within a region are identical, up to a scale factor. 
The Index Flood approach is not exempt from critics as its
application requires strong assumptions.  One major implicit
assumption, noticed by \citet{Gupta1994}, is that the coefficient of
variation of peak flows is to be constant across the region. This
fundamental property seems not to be verified in practice
\citep{Robinson1997} and not physically justified \citep{Katz2002}.
\par
The assumptions of the Index Flood model need often to be relaxed to
suit the observations. For this purpose, \citet{Gabriele1991} proposed
a hierarchical approach to RFFA\@. The skewness is still
supposed to be constant on the whole region, but the coefficient of
variation and the mean annual flood can vary slowly from one subregion
to another. However, the two authors underlined the practical
difficulty to delineate these subregions.
\par
In the Index Flood model, each observation from any site within the
region have the same weight. However, it seems not optimal as,
obviously, the most precious information come from the target
site. Indeed target data - even short - are the only one which are
``really'' distributed as the target site.
\par
We suggest here to carry out a Bayesian approach that encompasses the
classical Index Flood model and uses the whole data in a more
efficient manner.
In summary, the proposed
Bayesian approach differs from the Index Flood model as it: a) uses
the at-site information in a more efficient way since this approach
distinguishes the target site data and the regional data; b) does not
impose a purely deterministic relationship between sites within the
region.
\par
The main goal of this article is to test the efficiency and robustness
of the developed regional Bayesian model when dealing with short
record length series. For this purpose, classical frequency analysis
\textit{i.e.} local and traditional RFFA will be compared to the
suggested regional Bayesian approach. Section~\ref{sec:indexflow}
presents a brief summary of the classical Index Flood model. Relevant
theoretical aspects of Bayesian theory are introduced and applied to
flood modelling in a RFFA context in Section~\ref{sec:regbayes}.
Section~\ref{sec:data} describes the data set used to illustrate the
method. Section~\ref{sec:elic-prior} describes the procedure used to
elicit the prior distribution.  Section~\ref{sec:ClassApp} outlines
the weaknesses and strengths of each approach on a typical homogeneous
region. Finally section~\ref{sec:CompReg} presents an analysis of the
effect of homogeneity level on quantile estimation.

\section{The Index Flood model}
\label{sec:indexflow}

The Index Flood method states that flood frequency distributions
within a particular region are supposed to be identical when divided
by a scale factor - namely the Index Flood.  Mathematically, this
assumption is expressed as:
\begin{equation}
  \label{eq:sitdist}
  Q^{(S)} = C^{(S)} Q^{(R)}
\end{equation}
where $Q^{(S)}$ is the quantile function at site $S$ , $C^{(S)}$ the
Index Flood at site $S$ and $Q^{(R)}$ the regional quantile function
\textit{i.e.} the dimensionless quantile function valid across the
homogeneous region.
\par
Equation \eqref{eq:sitdist} is the core of the model and leads to
strong constraints concerning at-site distribution parameters.
Consequently, the shape parameter is the same throughout the
homogeneous region, whereas the location and scale parameters have
simple scaling behaviour - see Appendix~A.
\par
Equation \eqref{eq:sitdist} is supposed to be satisfied if all sites
are hydrologically and/or statistically similar. Therefore, one of the
main aspects of this approach is to identify a homogeneous region
which includes the target site.
\par
Similarity in basin characteristics is necessary but not sufficient to
ensure the homogeneity of the region regarding the statistics of the
flood peaks.  \citet{Hosking1993,Hosking1997} suggested a
heterogeneity measure $H_1$ to assess if a region is ``acceptably
homogeneous'' ($H_1 < 1$), ``probably heterogeneous'' ($1 \leq H_1 <
2$) or ``definitively heterogeneous'' ($H_1 \geq 2$). Note the case
$H_1 \leq 0$ seems to detect correlations between sites within the
region.
\par
Once the region satisfies the homogeneity test of
\citet{Hosking1993,Hosking1997}, the regional flood frequency
distribution and the related at-site distribution is computed in a
classical way. That is, by fitting the regional distribution to the
weighted mean of sample \textit{L}-moments. Details for computing
heterogeneity statistics, regional flood frequency and at-site
distribution can be found in \citep{Hosking1993,Hosking1997}.
\par
By definition of the Index Flood model, it can be seen that any
realisations of each samples have the same weight. Giving equal
weights to all site observations is debatable since the most relevant
information is certainly the target site one. Relevance of the target
site information is obvious as this is the only one which is
``really'' distributed as the target site. Thus, in this approach the
available information is not efficiently used.

\section{Regional Bayesian model}
\label{sec:regbayes}

The Bayesian concepts have already been applied with success to the
regional frequency analysis of extreme rainfalls \citep{Coles2003b}
and floods \citep{Madsen1997c, Northrop2004}. 
Regional information is not used to build a regional distribution 
but to specify a kind of
``suspicion'' about the target site distribution.  This is easily
achieved in the Bayesian framework through the so called prior
distribution.
\par

The main goal of Bayesian inference is to compute the posterior
distribution. The posterior distribution $\pi(\theta | x)$ is given by
the Bayes Theorem:
\begin{eqnarray}
  \notag
  \pi\left(\theta | x\right) &=& \frac{\pi\left(\theta\right) \pi\left(x
      ; \theta\right)}{\int_{\Theta}\pi\left(\theta\right) \pi\left(x
      ; \theta\right)d\theta }\\
  \label{eq:bayesthm} &\propto& \pi\left(\theta\right) \pi\left(x ;
    \theta\right)
\end{eqnarray}
where $\theta$ is the vector of parameters of the distribution to be
fitted, $\Theta$ is the space parameter. $\pi\left(x ; \theta\right)$
is the likelihood function, $x$ is the vector of observation and
$\pi\left(\theta\right)$ is the prior distribution.
\par
In theory, the posterior distribution is entirely known but is often
insolvable - because of the integral. One of the solutions is to
fix a prior model which leads to an analytical - or semi-analytical -
posterior distribution and which allows the posterior distribution to
be computed more easily \citep{Parent2003}. Nevertheless, the most
convenient way is to implement Markov Chain Monte Carlo
(\textbf{MCMC}) techniques to sample the posterior distribution. This
approach avoids using a purely artificial prior model with no
theoretical and/or physical justifications.
\par
For our application, the likelihood function corresponds to the GP
distribution as peaks over a threshold are of interest. From
Eq.~\eqref{eq:bayesthm}, if the prior distribution is known, posterior
distribution can be computed - up to a constant. The next section
describes how to define the prior distribution.

\subsection{Prior Distribution}
\label{subsec:priorDist}

The prior model is usually a multivariate distribution which must
represent beliefs about the distribution of the parameters
\textit{i.e.}  $\mu$, $\sigma$ and $\xi$ prior to having any
information about the data.
\par
As the proposed model is fully parametric, the prior distribution
$\pi(\theta)$ is a multivariate distribution entirely defined by its
hyper parameters. In our case study, the marginal prior distributions
were supposed to be independent lognormal for both location and scale
parameters and normal for the shape parameter. Thus,
\begin{equation}
  \label{eq:prior}
  \pi(\theta) \propto J \exp\left[ (\theta' - \gamma)^T \Sigma^{-1} (\theta'
    - \gamma) \right]
\end{equation}
where $\gamma, \Sigma$ are hyper parameters, $\theta' = (\log \mu,
\log \sigma, \xi)$ and $J$ is the Jacobian of the transformation from
$\theta'$ to $\theta$, namely $J=1/\mu\sigma$. $\gamma=(\gamma_1,
\gamma_2, \gamma_3)$ is the mean vector, $\Sigma$ is the covariance
matrix. As marginal priors are supposed to be independent, $\Sigma$ is
a 3-3 diagonal matrix with diagonal elements $d_1, d_2, d_3$.

\subsection{Estimation of the hyper parameters}
\label{subsec:estHypPar}

Hyper parameters are defined through the Index Flood concept. Consider
all sites of a region except the target site - say the $j$-th site. A
set of pseudo target site parameters can be computed:
\begin{eqnarray}
  \label{eq:pseudoloc}
  \tilde{\mu}^{(i)} &=& \mu_*^{(i)} C^{(j)}\\
  \label{eq:pseudoscale}
  \tilde{\sigma}^{(i)} &=& \sigma_*^{(i)} C^{(j)}\\
  \label{eq:pseudoshape}
  \tilde{\xi}^{(i)} &=& \xi_*^{(i)}
\end{eqnarray}
for all $i \neq j$, where $C^{(i)}$ is the at-site Index Flood and
$\mu_*^{(i)}, \sigma_*^{(i)}, \xi_*^{(i)}$ are respectively the
location, scale and shape at-site parameter estimates from rescaled
sample.
\par
Under the hypothesis of the Index Flood model, pseudo parameters
$\left(\tilde{\mu}^{(i)},
  \tilde{\sigma}^{(i)},\tilde{\xi}^{(i)}\right)$ for $i \neq j$ are
expected to be similar to the target site distribution parameters.
Note that, information from the target site sample is not used to
elicit the prior distribution. Thus, $C^{(j)}$ in equations
\eqref{eq:pseudoloc} and \eqref{eq:pseudoscale} must be estimated
without use of the $j$-th site sample.
\par
From these pseudo parameters, hyper parameters can be computed:
\begin{align}
  \label{eq:gamma1}
  \gamma_1 =& \frac{1}{N-1}\sum_{i\neq j} \log \tilde{\mu}^{(i)}, &d_1
  =& \frac{1}{N-1}\sum_{i \neq j} Var\left[ \log \tilde{\mu}^{(i)}
  \right] \\
  \label{eq:gamma2}
  \gamma_2 =& \frac{1}{N-1}\sum_{i\neq j} \log \tilde{\sigma}^{(i)}, &d_2
  =& \frac{1}{N-1}\sum_{i \neq j} Var\left[ \log \tilde{\sigma}^{(i)}
  \right] \\
  \label{eq:gamma3}
  \gamma_3 =& \frac{1}{N-1}\sum_{i\neq j} \tilde{\xi}^{(i)},  &d_3
  =& \frac{1}{N-2}\sum_{i \neq j} \left(\tilde{\xi}^{(i)} - \gamma_3
  \right)^2
\end{align}

It is important at this step to incorporate the uncertainties on the
elicitation of the prior distribution. Indeed, it may avoid problems
related to misleading information resulting from a region not so
homogeneous and moderating a ``suspicion'' that may be too true.
\par
For this purpose, two types of uncertainties are taken into account:
the one from parameter estimation, and the other one from target site
Index Flood estimation. Thus, hyper parameters $\gamma_1$ and
$\gamma_2$ are estimated differently than $\gamma_3$ as pseudo
parameters for location and scale parameters depends on the target
site Index Flood. Under the hypothesis of independence between
$C^{(j)}$ and $\mu_*^{(i)}, \sigma_*^{(i)}$ the variance terms in
Eq.~\eqref{eq:gamma1} and \eqref{eq:gamma2} are computed according
these two types of uncertainties:
\begin{align}
  \label{eq:var1}
  Var\left[ \log \tilde{\mu}^{(i)} \right] = Var\left[ \log C^{(j)}
  \right] + Var\left[\log \mu_*^{(i)} \right]\\
  \label{eq:var2}
  Var\left[ \log \tilde{\sigma}^{(i)} \right] = Var\left[ \log C^{(j)}
  \right] + Var\left[\log \sigma_*^{(i)} \right]
\end{align}
The independence assumption between $C^{(j)}$ and $\mu_*^{(i)}$,
$\sigma_*^{(i)}$ is not too restrictive as the target site Index Flood
$C^{(j)}$ is estimated independently from $\mu_*^{(i)}$,
$\sigma_*^{(i)}$.
\par
Note that $Var\left[ \log \cdot_*^{(i)} \right]$ are estimated thanks
to Fisher Information and the Delta method. Estimation of
$\mathrm{Var}\left[ \log C^{(j)}\right]$ is a special case and depends
on the method for estimating the at-site Index Flood.  Nevertheless,
it is always possible to carry out an estimation of this variance, at
least through standard errors.

\subsection{Specificities of the proposed prior model}
\label{subsec:specPrior}

The construction of the prior distribution with the regional
information was already suggested by \citep{Northrop2004}.
Nevertheless, the location parameter - or equivalently the threshold
in the GP case - was supposed to be known. Yet, from a theoretical
point of view, the location parameter can not be known prior to having
any information from the target site sample. \citet{Northrop2004}
developed a similar approach based on the Index Flood but uncertainty
associated with the scale factor prediction was not considered. The
prior distribution was elicited directly from the distribution of the
``pseudo target site'' estimates $(\tilde{\mu}^{(i)},
\tilde{\sigma}^{(i)}, \tilde{\xi}^{(i)})$.  In this perspective,
``pseudo target site'' estimates are viewed as constant and not as
random variables. When dealing with sites with a long record,
uncertainties on parameter distributions are low. On the contrary,
this have a much more impact for the Index Flood as uncertainties are
as much larger as the target site Index Flood is estimated without use
of at-site data - even with long record length sites. Note that if the
prior distribution is overly accurate, estimation and credibility
intervals are influenced. For these reasons and unlike the approach
proposed by \citet{Northrop2004}, the target site Index Flood in the
proposed methodology is considered to be a random variable and not a
constant.
\par
Thus, our prior distribution is not too falsely ``tight fit''. But it
reflects ``real'' beliefs about target site behaviour without any use
of target site sample.
\par
\citet{Madsen1997c} and \citet{Fill1998} both presented a regional
empirical Bayesian estimator. Both models used conjugates families for
prior distributions. However, even if conjugates families are
convenient devices, they should not only be used just because
computations are easier. In their approaches, prior distributions are
elicited with quantile regression on relevant physiographic
characteristics.
\par
Our approach differs differs from the two previous empirical Bayesian
approaches (i.e target site sample is not used to elicit priors) and
respects in that way absolutely the Bayesian theory. Moreover,
conjugate priors are not considered, but priors are suited to the
data. For example, the lognormal distribution for both location and
shape parameters is justified by a physical and theoretical lower
bound as: a) discharge data are naturally non negative; so the
location parameter should also be non negative; b) the scale parameter
is strictly positive by definition of the GP distribution.
\par
This prior model is quite different from the one proposed by
\citet{Coles1996} who introduced a lognormal prior distribution only
for the scale parameter. Note that it is possible to work with return
levels \citep{Coles1996} or return periods \citep{Crowder1992} instead
of working with distribution parameters. However, regional information
is suited to work directly with distribution parameters. For other
studies, such prior models could be of interest if ``suspicion'' is
based on return levels or return periods.

\section{Data description}
\label{sec:data}

Streamflow data were collected at 48 gauging stations in an area
reaching from the $45^{th}$ to $47^{th}$ N and from the $3^{rd}$ to
the $8^{th}$ E. The selection of the gauging sites was initially based
on the 22 regions into which France is divided for the implementation
of the Water Framework Directive \citep{Wasson2004}. Seven regions
cover the area under study. These regions were manually delineated
taking into account the spatial pattern of mean annual rainfall,
elevation and underlying geology. All these variables might influence
flood generation processes. Therefore this division is considered as a
preliminary guide for pooling stations. According to
\citet{Hosking1997}, pre-regions were slightly altered by identifying
discordant sites while maximising the number of site within the region
and meeting the heterogeneity test. Finally, a set of 14 stations was
selected for this study.  The heterogeneity statistic for this group
is $H_1 = 0.17 < 1$. Consequently, the region be considered as
``acceptably homogeneous''.
\par
The dataset includes seven tributaries to the Loire River and seven
gauging stations located in the French part of the Rh\^one basin
(Fig.~\ref{fig:reghomo}, Tab.~\ref{tab:reghomo}). The record length of
the instantaneous discharge time series ranges from a minimum of 22
years to a maximum of 37 years, with a mean value of 32 years. The
drainage areas vary from 32 to 792 km$^2$.  Moreover, most of the
gauging stations monitored first-order stream catchments \textit{i.e.}
all but two pairs of catchments are unnested. The large majorities of
flood in the region occur during autumn and winter and are caused by
heavy liquid precipitation.
\begin{figure}
  \centering
  \includegraphics[width=0.5\textwidth]{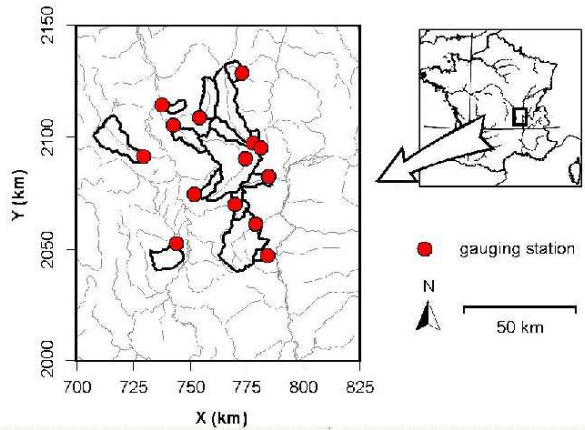}
  \caption{Location of the gauging stations within the studied area}
  \label{fig:reghomo}
\end{figure}
\begin{table*}
  \centering
  \caption{Characteristics of the stations of the homogeneous region}
  \begin{tabular}{llcccc}
    \hline
    Code & Station & Area $\left(km^2\right)$ & X $\left(km\right)$ & Y $\left(km\right)$ & Record\\
    \hline
    K0624510 &  The Bonson river at St Marcellin  & 104 & 744.72 & 2053.90 & 1971-2003\\
    K0663310 &  The Coise river at Larajasse & 61 & 770.67 & 2072.11 & 1971-2003\\
    K0704510 &  The Toranche river at St Cyr & 62.3 & 752.63 & 2076.68 & 1977-2003\\
    K0813020 &  The Aix river at St Germain Laval  & 193 & 729.48 & 2093.71 & 1973-2002\\
    K0943010 &  The Rhins river at Amplepuis & 114 & 754.52 & 2111.10 & 1973-2003\\
    K0974010 &  The Gand river at Neaux & 85 & 743.45 & 2107.75 & 1972-2003\\
    K1004510 &  The Rhodon river at Perreux & 32 & 738.40 & 2116.64 & 1973-2003\\
    U4505010 &  The Ardieres river at Beaujeu & 54.5 & 773.67 & 2130.75 & 1969-2003\\
    U4624010 &  The Azergues river at Chatillon & 336 & 779.07 & 2099.72 & 1970-2003\\
    U4635010 &  The Brevenne river at Sain Bel & 219 & 775.90 & 2092.57 & 1969-2003\\
    U4644010 &  The Azergues river at Lozanne & 792 & 782.56 & 2098.09 & 1981-2003\\
    V3015010 &  The Yzeron river at Craponne & 48 & 785.47 & 2084.50 & 1969-2003\\
    V3114010 &  The Gier river at Rive de Gier & 319 & 780.54 & 2062.67 & 1981-2003\\
    V3315010 &  The Valencize river at Chavanay & 36 & 786.54 & 2048.60 &
    1978-2003\\
    \hline
  \end{tabular}
  \label{tab:reghomo}
\end{table*}
\par

Partial duration flood series were extracted from the time series for
each station. Fig.~\ref{fig:timeseries} illustrates time series for
stations U4505010, U4635010 and V3015010 and their associated
thresholds. Threshold levels were selected to extract in average around
two events per year while meeting the criteria of independence between
floods \citep{Lang1999}.
\par

Three stations U4505010, U4635010 and V3015010 were of particular
interest because of their extended record length of 37 years. The time
series of those three sites are displayed in Fig.~\ref{fig:timeseries}. 
\begin{figure*}
  \centering
  \includegraphics[width=1\textwidth]{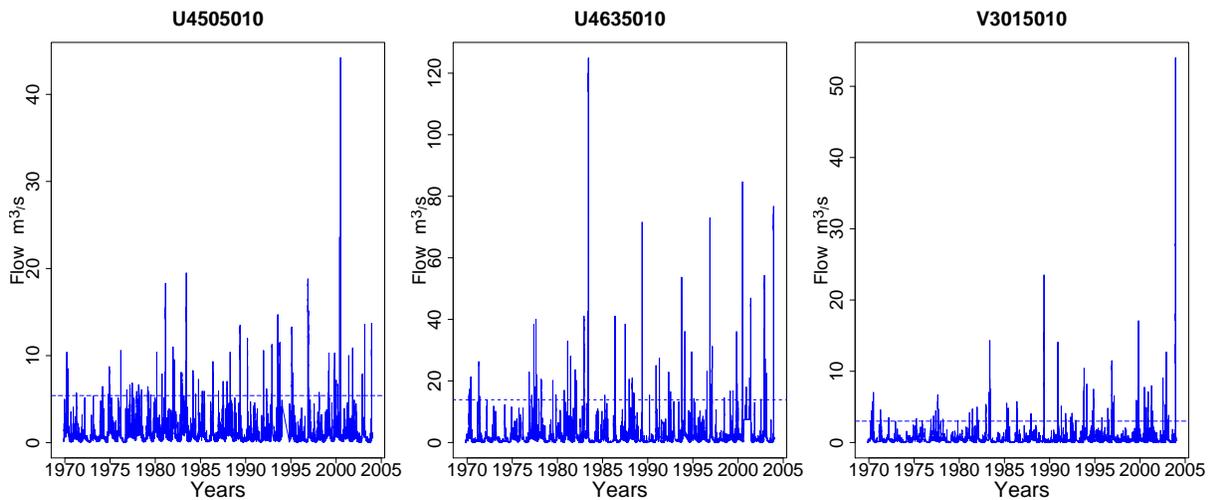}
  \caption{Times series for sites U4505010, U4635010
    and V3015010 and thresholds associated}
  \label{fig:timeseries}
\end{figure*}
\par

In this case study, the scale factor was set to correspond to the
1-year return flood quantile - or equivalently the quantile associated
with probability of non exceedance 0.5. Thus, our choice for the Index
Flood is close to the sample median which was the reference in
\citet{FEH1999} but differs from \citet{Hosking1997} where the sample
mean was used.  This particular choice for the Index Flood is not
unintentional as estimating the quantile with probability of non
exceedance 0.5 is more robust than estimating the sample
mean. Analysing the influence of Index Flood selection is beyond the
scope of this work. The main point is to keep the same Index Flood
throughout the case study to compare approaches on the same basis.

\section{Elicitation of the prior distribution}
\label{sec:elic-prior} 

To estimate the target site Index Flood, the most popular way is 
to develop an empirical formula that relates the 
flow statistic to geomorphological, land-use
and climatic descriptors. This relationship is usually established by
multivariate regression procedures. In our case study, we consider a
simple model for which only one explanatory variable is introduced in
the regression analysis: the drainage area. The power form model is
adopted:
\begin{equation}
  \label{eq:crupedix}
  C = a A^b
\end{equation}
where $A$ is the area of the catchment. Parameters $a$ and $b$ are
computed through ordinary least square procedures on logarithmically
transformed data.
\par
However, more sophisticated models could be carried out. Nevertheless,
for our case study, observations demonstrate that
Eq.~\eqref{eq:crupedix} is a good parametrisation for estimating the
Index Flood. Fig.~\ref{fig:priorijust} and~\ref{fig:regresscrupedix}
illustrate the efficiency of the regressive and prior model for site
U4505010 for which:
\begin{equation}
  \label{eq:estimindflow2}
  \hat{C} = 0.12 A^{1.01}\hspace{1cm} \left(R^2 = 0.94 \right)
\end{equation}
\begin{figure*}
  \centering
  \includegraphics[angle=-90,width=1\textwidth]{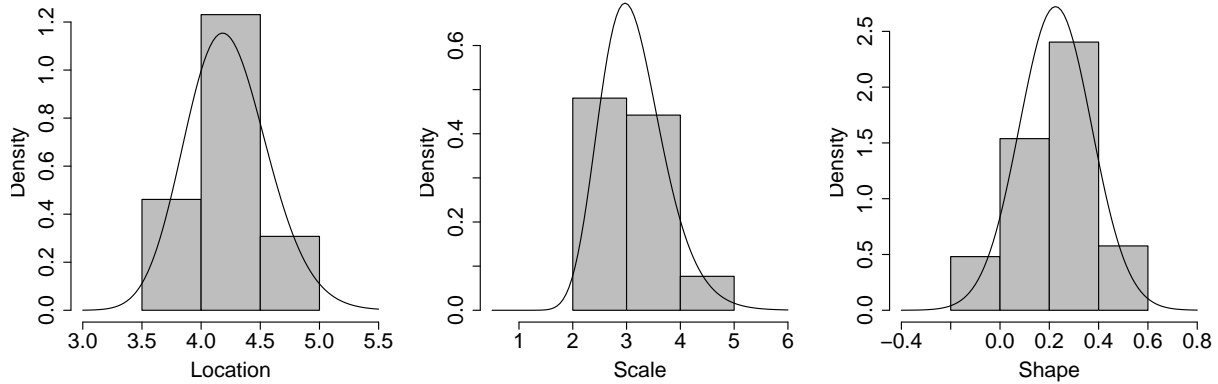}
  \caption{Histograms of pseudo target site estimates of location, scale
    and shape parameters for site U4505010}
  \label{fig:priorijust}
\end{figure*}
\begin{figure}
  \centering
  \includegraphics[angle=-90,width=0.5\textwidth]{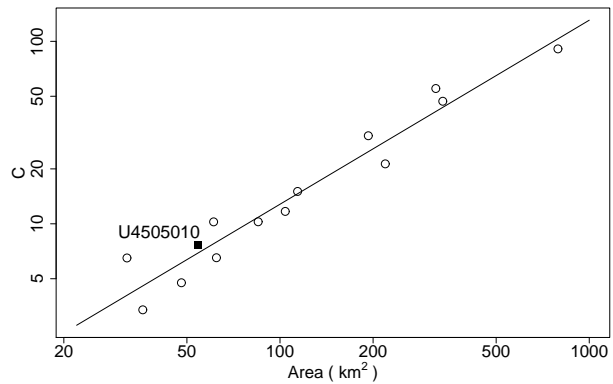}
  \caption{Regression on the basin area for estimating the at-site
    index flow for station U4505010}
  \label{fig:regresscrupedix}
\end{figure}
\par
Regional information was incorporated in the prior
distribution through the Index Flood model. Moreover, uncertainties in
the prior distribution were incorporated. Thus, the prior information
is, on one hand, not too falsely accurate and on the other hand,
informative enough because of the supposed homogeneity of stations.

\section{Performance of the Bayesian model on a homogeneous region}
\label{sec:ClassApp}

When making classical inference on small samples, the uncertainties
may be too large. If an extremal event or too many ``regular events''
in this short record period are present, estimation can be impacted.
It could lead to a dramatic overestimation or underestimation of
quantiles corresponding to different return levels. A perfect model is
expected~: a) to perform well enough even with small samples; b) to be
robust enough when an extreme event occurs in the sample; c) to be
robust enough when too many ``regular'' events occur in the sample.
\par

In this section, three different models will be applied. For this
purpose, the three stations U4505010, U4635010 and V3015010 were
selected to assess the robustness and efficiency of the local,
regional and Bayesian regional models. These three different
approaches correspond to~: a) local: fit the GP distribution to the
peaks over threshold data with the Maximum Likelihood Estimator (MLE),
Unbiased Probability Weighted Moments (PWU) and the Biased Probability
Weighted Moments (PWB); b) regional (REG): fit a regional GP
distribution as described in section~\ref{sec:indexflow} and obtain
the target site distribution; c) regional Bayesian (BAY): elicit the
prior density from regional information, then compute the posterior
density through MCMC techniques. As an illustration of MCMC output,
Fig.~\ref{fig:priorPost} displays the prior and posterior marginal
densities for the GP parameters of the proposed model. Marginal
posterior distribution obtained from an uninformative prior model are
also displayed. That is with the same prior model but with a large
variance - i.e. $d_i = 1000, i=1\ldots3$.
\begin{figure*}
  \centering
  \includegraphics[angle=-90,width=1\textwidth]{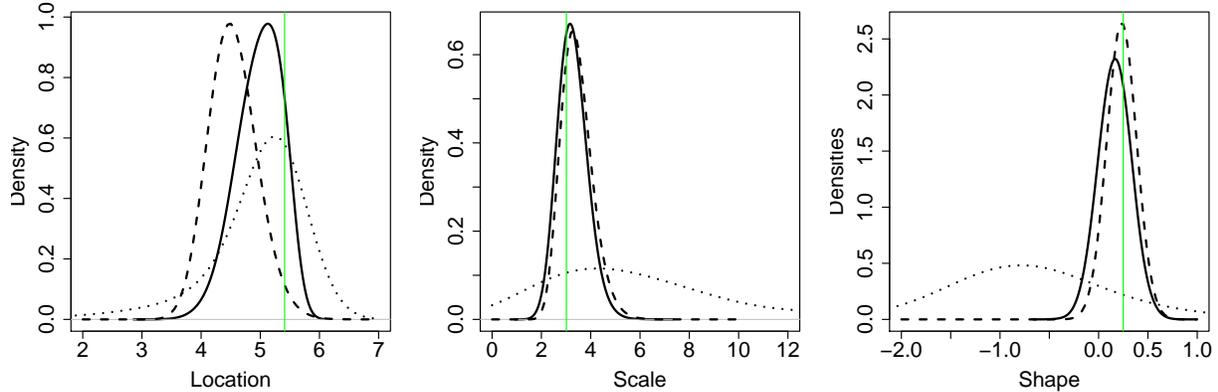}
  \caption{Proposed prior (dashed line), proposed posterior (solid
    line) and posterior from an uninformative prior (dotted line)
    marginal densities for GP parameters. Site U4505010 with 5 years
    record length. Vertical lines denotes benchmark values.}
  \label{fig:priorPost}
\end{figure*}
Fig.~\ref{fig:priorPost} shows the relevance of regional information
as the proposed prior model is clearly more accurate than an analysis
directly from data. Moreover, for the proposed model and even with
only 5 years record length, marginal posterior densities are more
accurate than marginal prior densities - except for the shape
parameter. Thus, combination of regional and target site information
at two different stages is worthwhile, even when only few data are
available. Location parameter is a special case as the modes of both
marginal prior and posterior densities seem to be significantly
dissimilar.
\par
As the main goal of this work is to compare models on small samples,
efficiency will be evaluated on sub-samples from the original data.
Local Maximum Likelihood Estimation on the whole sample will be used
as benchmark to assess the performance of each model. This particular
case will be denoted THEO in the following sections. The choice of MLE
estimate as a benchmark value is reasonable because of its theoretical
motivation and asymptotic efficiency. Moreover, the MLE approach
allows the calculation of profile confidence intervals.  This is a key
point as these profile confidence intervals are often more accurate
than those based on the Delta Method and Fisher Information
\citep{Coles2001}.
\par
Furthermore, as interpretation on quantile estimates is more natural
than for distribution parameter estimates, the analysis will focus on
quantiles corresponding to return period 2, 5, 10 and 20 years.
Benchmark values for these quantiles - and their associated 90\%
profile likelihood confidence intervals are detailed in
Tab.~\ref{tab:BenchQuant}.  Benchmark values with return periods
greater than 20 years will be considered unreliable - as uncertainties
on these quantiles are too large with only 37 years of record.
\begin{table*}
  \centering
  \caption{Benchmark values for 2, 5, 10 and 20 years quantiles and the
    associated 90\% profile likelihood confidence intervals in bracket}
  \begin{tabular}{lcccc}
    \hline
    Station & $Q_2$ & $Q_5$ & $Q_{10}$ & $Q_{20}$\\
    \hline
    U4505010 & 10.8\hfill (10.1, 11.7) & 15.3\hfill (13.9, 17.4) &
    19.5\hfill (17.2, 23.4) & 24.4\hfill (20.6, \ 31.5)\\
    U4635010 & 33.0\hfill (30.0, 36.5) & 52.2\hfill (45.5,
    62.5) & 72.2\hfill (60.2, 95.4)& 98.9\hfill (69.2, 200.5)\\
    V3015010 & \ 7.5\hfill (\ \ 6.9, \ 8.3)& 11.7 \hfill (10.4,
    13.7) & 15.9\hfill (13.6, 19.9)& 21.3\hfill (17.3, \ 28.8)\\
    \hline
  \end{tabular}
  \label{tab:BenchQuant}
\end{table*}
\par
Moreover, for such return periods, benchmark values are quite
equivalent to those obtained with PWM estimates - with a mean bias of
0.89\%. So, performance of each model is not too much impacted by the
choice of the MLE estimator for benchmark values.
\begin{figure*}
  \centering
  \includegraphics[width=1\textwidth]{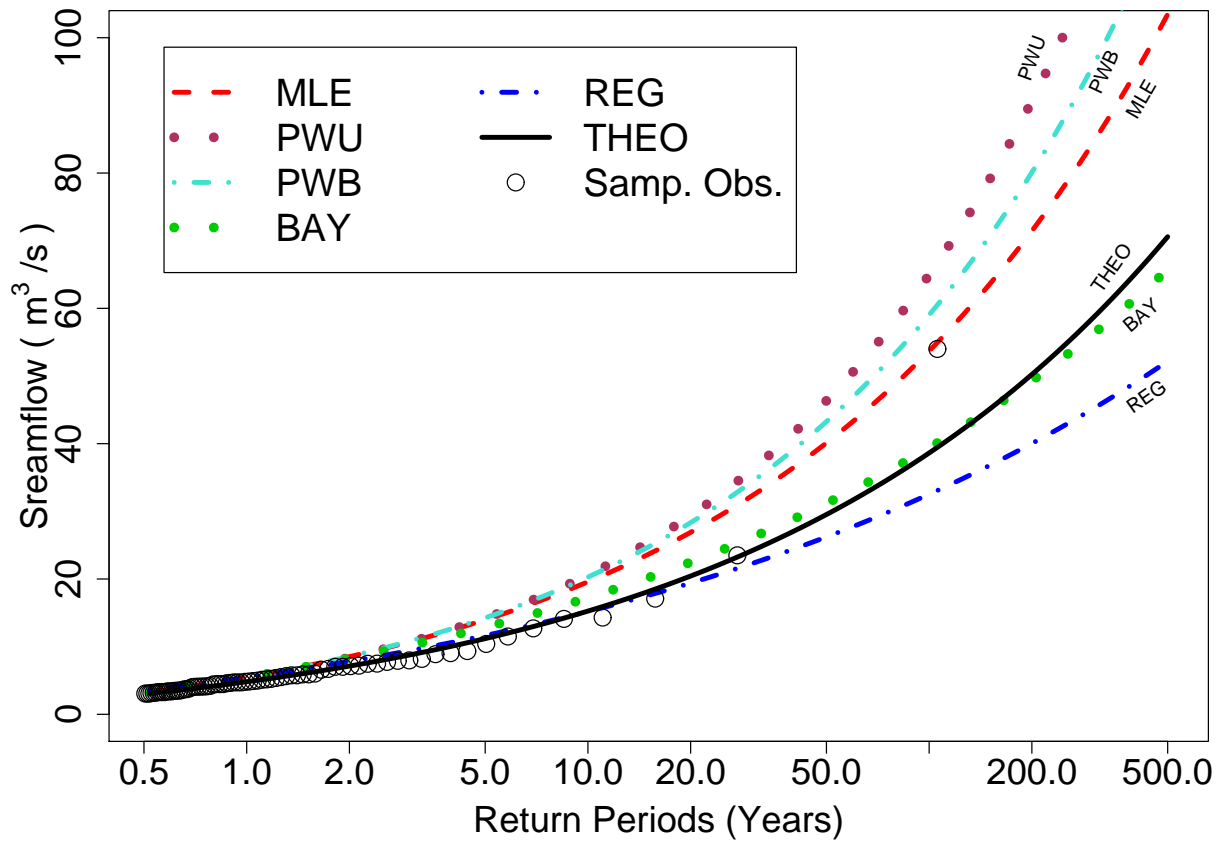}
  \caption{Comparison of frequency curves for site V3015010 with 15
    last years recording}
  \label{fig:retlev}
\end{figure*}
\par
Different frequency curves for site V3015010 with only the last 15
years recording are displayed in Fig.~\ref{fig:retlev}.  Let us focus
on the largest observation.  Return period related to this event is
very high for the REG approach.  All other models lead to
significantly lower return periods. This flood event is extreme at
regional scale but not anymore in a local context.  This
underestimation is due to the misuse of the target site sample to
establish the regional distribution. On the other side, the regional
Bayesian model performs well for all return periods.  Indeed,
Fig.~\ref{fig:retlev} indicates that the return level curve is very
similar to benchmark one.  This is quite logical as it adds up the
advantage of using efficiently the target site sample and a good
``suspicion'' on the global behaviour of the flood peak distribution
thanks to the so-called prior distribution.  Local approaches suggest
a very heavy tail as the extremal event of year 2004 (see
Fig.~\ref{fig:timeseries}, right panel) was in the last sequence of 15
years of records.
\par

As one of the main goals of a RFFA procedure is to deal with small
samples, the target site sample was truncated to obtain shorten
periods of records of $m$ years, $m \in
\left\{5,10,15,20,25,30,37\right\}$.  Robustness and efficiency of the
methods to converge to the parameters of the target site distribution
are measured. For this purpose, quantile estimates corresponding to
return period 2, 5, 10, 20 years - corresponding to non-exceedance
probabilities 0.75, 0.9, 0.95 and 0.975 respectively - are pointed.
The evolution of quantile estimates as a function of the record length
period is presented in Fig.~\ref{fig:quantsize}. The figure is
achieved considering only the first $m$ years ; that is, for example,
estimates related to the 5-year record length corresponds to the
period 1969 -- 1973.
\begin{figure*}
  \centering
  \includegraphics[width=1\textwidth]{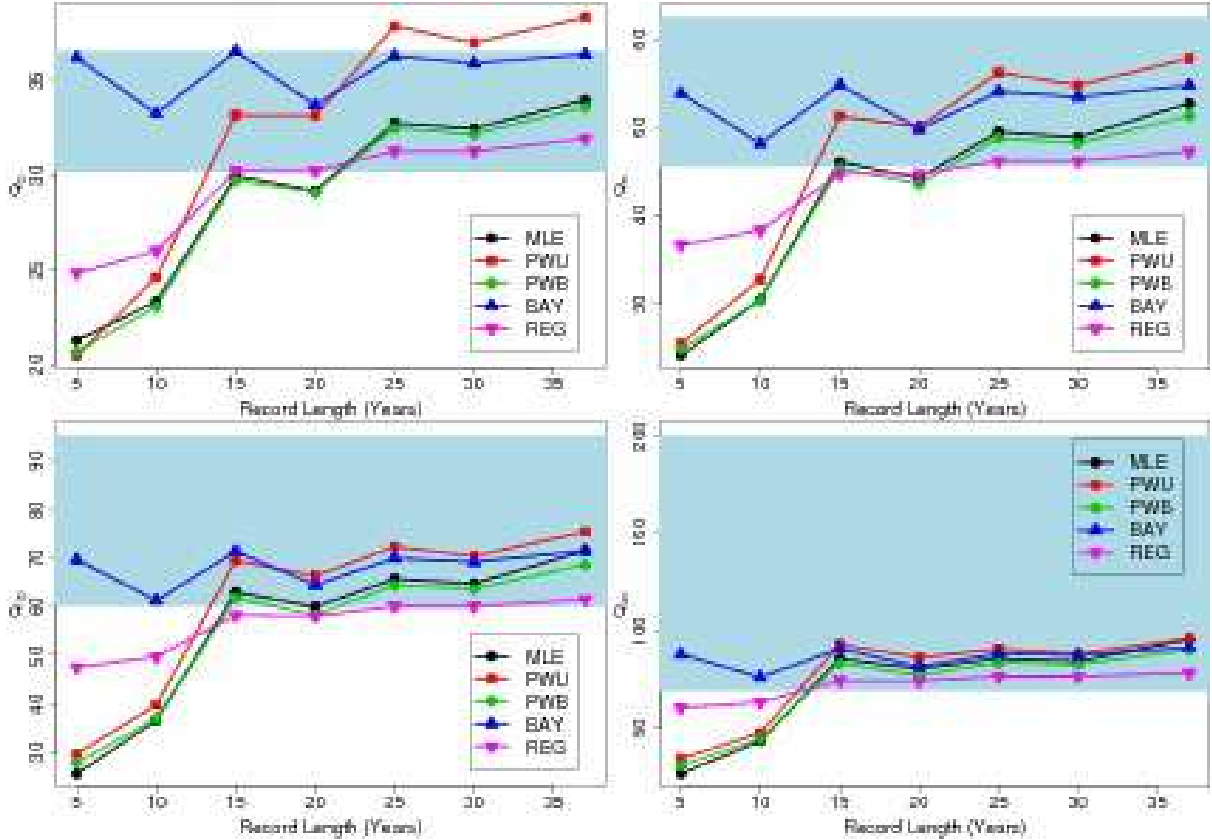}
  \caption{Evolution of $Q_2, Q_5, Q_{10}, Q_{20}$ estimates as the
    size increases for the site U4635010 and 90\% profile likelihood
    confidence interval for the benchmark values - light blue area}
  \label{fig:quantsize}
\end{figure*}
\par

Because of the extreme event observed in 1983 (see
Fig.~\ref{fig:timeseries}, middle panel), systematic underestimation
of benchmark values for local and REG approaches can be noticed. This
result shows that: on one hand, for small samples classical inference like
MLE, PWB and PWU are too responsive if too many ``regular'' events
occurred. On the other hand, for the Index Flood model,
underestimation of quantiles is related to the underestimation of the
scale factor $C^{(S)}$ in Eq.~\eqref{eq:sitdist} because of these
``regular'' events. Only the Bayesian model performs well enough even
with record lengths lower than 15 years.
\par
A monotonic increase of the design flood estimates with the sample
length can be noticed in Fig.~\ref{fig:quantsize}. This behaviour is
easily explained by Fig.~\ref{fig:timeseries}, middle panel. Indeed,
only the last part of the time series shows really extreme events. As
the record length increases, much more extreme events occur leading to
higher estimates. The Bayesian approach is the only one which does not
really present this monotonic behaviour.
\par
Moreover, the Bayesian approach is by far the most robust and accurate
model as, on the whole range of record length, and for all benchmark
values, estimation lies in the 90\% profile likelihood confidence
interval.  This is not true with any other model.  The advantage of
incorporating regional information within a Bayesian framework is
certainly to define a ``restricted space'' where distribution
parameters belong to. Thus, the impact of a very extremal event - or
conversely too many low-level events - should be regarded as an
extreme event related to this ``restricted space''.
\par

The gain of accuracy in the target site from using regional
information is clearly established in section~\ref{sec:ClassApp}
(Fig.~\ref{fig:retlev} and~\ref{fig:quantsize}). The Bayesian approach
seems to be robust even with small samples while being accurate with
larger sample. The poor performance of the REG model is related to a
bad selection of sites within the ``homogeneous'' region being
considered and estimates may be more accurate if ``better'' regions
were considered.  Unfortunately, building up such region is difficult
because of the purely deterministic relation \eqref{eq:sitdist}. As
the Bayesian approach relaxes the REG model, the search for more
homogeneous regions could be ineffective. The goal of the next section
is to measure the potential gain, for the Bayesian model, against
homogeneity property.

\section{Effect of heterogeneity degree on quantile estimation}
\label{sec:CompReg}

As indicated in the previous section, we focus now on the impact of
the level of homogeneousness of the region. For this purpose, we
consider four different regions - denoted $He^+, He, Ho$ and $Ho^+$ -
which correspond to increasingly homogeneous regions according to the
test of \citet{Hosking1997}. The $Ho$ region corresponds to the region
analysed in the previous section and described in
Tab.~\ref{tab:reghomo}. All regions have 14 site except for the most
homogeneous one $Ho^+$ which contains only 8 stations. $He$ and $He^+$
regions are derived form $Ho$. One to five sites are withdrawn and
replaced by other stations to obtain larger heterogeneity measure. The
$Ho^+$ region is a sub-region of $Ho$. Heterogeneity statistics for
these regions are summarised in Tab.~\ref{tab:HetStat}.
\begin{table*}
  \centering
  \caption{Heterogeneity statistics for the four region considered -
    statistics in bracket are obtained with the scale factor taken to
    be the 1-year quantile corresponding to non-exceedance probability 0.5}
  \begin{tabular}{lcccc}
    \hline
    Region & $He^+$  & $He$ & $Ho$ &  $Ho^+$\\
    \hline
    $H_1$ & 7.11 \hfill   (6.83) & 1.35  \hfill   (1.37) & 0.17 \hfill
    ( 0.08) & -0.60 \hfill (-0.67)\\
    $H_2$ & 3.46 \hfill    (3.38) & 1.00 \hfill    (1.03) & 0.41
    \hfill  ( 0.33) & -1.28 \hfill (-1.31)\\
    $H_3$ & 1.40  \hfill   (1.45) & 0.30 \hfill    (0.28) & -0.09
    \hfill  (-0.14) & -1.14 \hfill (-1.18)\\
    \hline
  \end{tabular}
  \label{tab:HetStat}
\end{table*}
\par

To evaluate the influence of homogeneousness level of a region on
quantile estimation, models are assessed using two performance
criteria~: the Normalised Bias ($\mathbf{NBIAS}$) and the Normalised
Root Mean Squared Error ($\mathbf{NRMSE}$). These indices are defined
as follows~:
\begin{eqnarray}
  NBIAS &=& \frac{1}{k}\sum_{i=1}^k \frac{\hat{Q}_i - Q}{Q}\\
  NRMSE &=& \sqrt{\frac{1}{k}\sum_{i=1}^k \left(\frac{\hat{Q}_i-Q}{Q}\right)^2}
\end{eqnarray}
where $k$ is the number of estimates of $Q$ and $\hat{Q}_i$ is the
i-th estimate of the benchmark value $Q$. To compute these two
indices, we fit all models on all trimmed periods of size $m$ years -
$m \in \left\{5, 10, 15, 20, 25, 30\right\}$. Moreover, the overall
performance of each model is evaluated using a rank score.  This
technique was already used to compare different models in
\citet{Shu2004}.
\par

To calculate the rank score, the $p$ models are ordered thanks to
their performance indices - 1 corresponding to the best model and $p$
to the worst. For each model, the scores for the different criteria
are summed to obtain the overall rank score $R_o$ for the model. For
convenience, the overall rank score $R_o$ is standardised in such a
way that in lies within the interval $\left[0,1\right]$:
\begin{equation}
  \label{eq:StdRank}
  R_s = \frac{pq - R_o}{pq - q}
\end{equation}
where $p$ is the number of models being considered, and $q$ the number
of indices. A standardised rank score close to 1 -\textit{resp.} 0 -
is associated to a model with a good - \textit{resp.} poor -
performance.
\par

Three quantiles are of particular interest $Q_5, Q_{10}$ and $Q_{20}$
- \textit{i.e.} associated to probability of non-exceedance 0.9, 0.95
and 0.975 respectively. $NRMSE$, $NBIAS$ and the standardised rank
score for station U4635010 and a record length of 5 years are
illustrated in Tab.~\ref{tab:Efftab}. Notations for different models
in this table consist of one lowercase letter referring to the
Bayesian approach $b$ or Regional Index Flood $r$ and the denomination
of the homogeneity degree of the region. Only the $MLE$ model does not
use these notations as it is completely independent of the homogeneity
level.
\begin{table*}
  \centering
  \caption{Estimation of  $NRMSE$ and $NBIAS$ for station U4635010
    with a record length of 5 years}
  \begin{tabular}{lcccrrrrrc}
    \hline
    \multirow{2}*{Model} & & $NRMSE$ & & & & $NBIAS$ & & &
    \multirow{2}*{Rank Score}\\
    \cline{2-4} \cline{6-8}
     &  $Q_5$ & $Q_{10}$ & $Q_{20}$ && $Q_5$ &
    $Q_{10}$ & $Q_{20}$ && \\
    \hline
    $MLE$ & 0.33 & 0.34 & 0.39 && 0.01 & $-$0.09 & $-$0.18 && 0.26 \\
    $bHe^+$ & 0.16 & 0.13 & 0.18 && 0.09 & $-$0.02 & $-$0.13 && 0.65 \\
    $rHe^+$ & 0.27 & 0.30 & 0.37 && $-$0.12 & $-$0.22 & $-$0.31 && 0.18 \\
    $bHe$ & 0.10 & 0.07 & 0.11 && 0.08 & 0.00 & $-$0.09 && 0.85 \\
    $rHe$ & 0.27 & 0.26 & 0.28 && $-$0.03 & $-$0.10 & $-$0.17 && 0.43 \\
    $bHo$ & 0.14 & 0.09 & 0.08 && 0.12 & 0.05 & $-$0.02 && 0.76 \\
    $rHo$ & 0.27 & 0.26 & 0.27 && 0.01 & $-$0.06 & $-$0.12 && 0.58 \\
    $bHo^+$ & 0.29 & 0.28 & 0.25 && 0.29 & 0.27 & 0.25 && 0.19 \\
    $rHo^+$ & 0.28 & 0.27 & 0.26 && 0.02 & $-$0.01 & $-$0.04 && 0.60 \\
    \hline
  \end{tabular}
  \label{tab:Efftab}
\end{table*}
\par

Results from Tab.~\ref{tab:Efftab} demonstrate that the Bayesian model
performs quite well independently of the region being considered.
However, this model seems to perform even better when applied to a
``acceptably homogeneous'' or ``probably heterogeneous'' region.  For
the $Ho^+$ region, the Bayesian approach performs poorly. This may be
explained by the fact that the prior distribution is too informative
and probably not consistent with the target site sample. This comment
is yet not discrepant with the good overall performance of the REG
model on this region.  Indeed, as the prior distribution is elicited
using equations \eqref{eq:pseudoloc}--\eqref{eq:pseudoshape}, and the
scale factor $C^{(j)}$ is estimated without any use of the target site
sample, this can lead to a misleading prior distribution while the REG
model performs well. The bad estimation of the scale factor is less
important with a more heterogeneous region as the prior information is
less informative, thus the Bayesian model performance is not highly
impacted.
\par

On the other side, the overall rank score of the REG model increases
with the homogeneity degree of the region. Yet, the overall rank score
for the REG model never exceeds the value of 0.6 - reached for the
$Ho^+$ region.  This value remains much lower than the best rank score
for the Bayesian model - \textit{i.e.} 0.85. These results corroborate
the superiority of the Bayesian approach.
\par

From Tab.~\ref{tab:Efftab}, two conclusions can be established. On one
hand, for small samples, the Bayesian approach is the most competitive
model. On the other hand, results seem to indicate that there is no
need to keep increasing the homogeneousness of the region as it
increases the risk of being too confident in the ``homogeneous
region'' without increasing significantly the efficiency of the model.
\par

These results are in line with similar results obtained for stations
U4505010 and V3015010, except for the bad behaviour of the Bayesian
model on the $Ho^+$ region. Indeed, for the other stations, the
Bayesian model remains more efficient than the REG model within the
$Ho^+$ region.  However, its overall rank score remains stable through
out the different region - $He, Ho$ and $Ho^+$. The ``risk'' to deal
with too much homogeneous region - as $Ho^+$ - is also corroborated as
the overall rank score for the Index Flood model for station U4505010
decreases dramatically until 0.06. Thus, the Index Flood for the
$Ho^+$ region performs quite well for stations U4635010 and V3015010,
while very surprisingly badly with U4505010.
\par

\begin{figure*}
  \centering
  \includegraphics[width=1\textwidth]{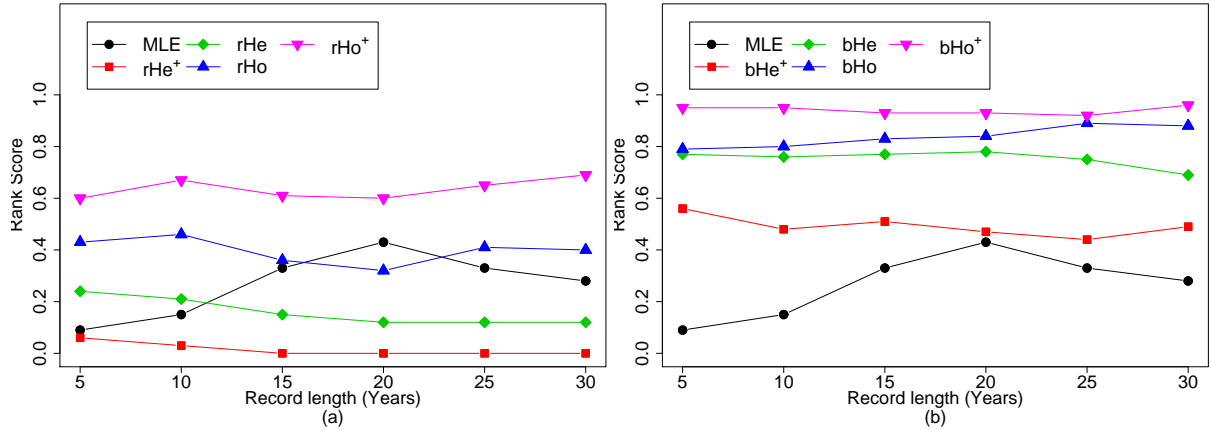}
  \caption{Score evolution as a function of record length for Station
    V3015010.  (a) REG scores, (b) BAY scores}
  \label{fig:ScoreVsSize}
\end{figure*}
In Fig.~\ref{fig:ScoreVsSize}, the evolution of the overall rank score
as a function of the record length is illustrated for station
V3015010. The left panel corresponds to the REG part, while the right
one stands for the Bayesian approach. In both panel, the $MLE$ score
is also presented. Fig.~\ref{fig:ScoreVsSize} indicates that the
evolution of the overall rank score is more stable for regional models
- that is REG and Bayesian models - than for the $MLE$.  Furthermore,
the benefit of increasing the homogeneity degree of the region is more
relevant for the REG model than for the Bayesian model.  Nevertheless,
the worst Bayesian rank score is always quite close to the best REG
rank score. This seems to indicate the superiority of the Bayesian
approach. This last point is corroborated with the results
corresponding to stations U4505010 and U4635010 except for the $bHo^+$
model for station U4635010 because of the bad estimation of the scale
factor $C^{(j)}$ - as denoted earlier. The effect of bad estimation of
the target site Index Flood on prior and thereby on posterior
distributions is depicted in Fig.~\ref{fig:badEstIF}.
\begin{figure*}
  \centering
  \includegraphics[angle=-90,width=1\textwidth]{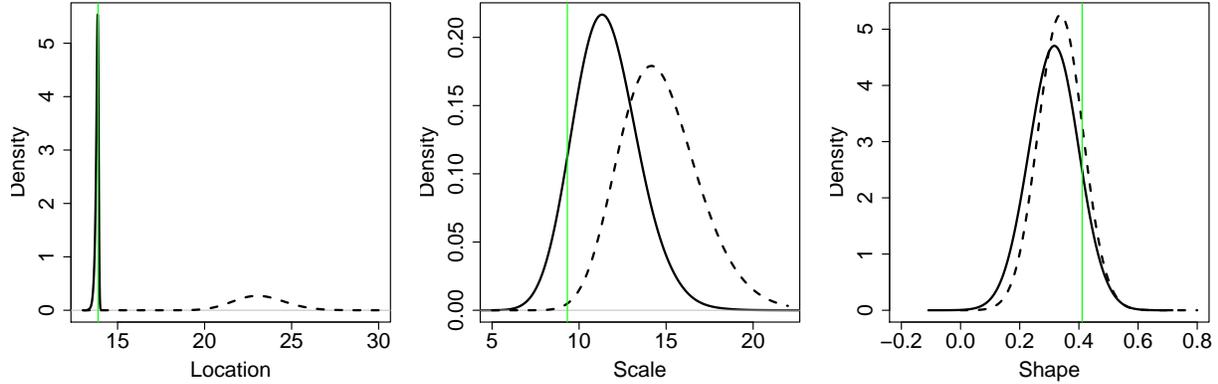}
  \caption{Effect of bad estimation of target site Index Flood on
    marginal prior and posterior densities. Site U4635010 with 10
    years record length.} 
  \label{fig:badEstIF}
\end{figure*}
From Fig.~\ref{fig:badEstIF}, it is overwhelming that the prior model
is not appropriate - particularly for the location parameter. Prior
for the shape parameter is not too false as it does not depend on the
target site Index Flood estimate.
\par
As the record length increases, the $MLE$ model becomes more and more
efficient. In particular, for record lengths greater than 15 years, it
is more effective than $rHe^+, rHe$ and $rHo$ models. On one hand, for
record lengths smaller than 15 years, $MLE$ is always less efficient
than Bayesian approaches and even significantly for $bHe, bHo$ and
$bHo^+$ models. This is quite logical as Bayesian estimation can be
looked at as a restrictive maximum likelihood estimator - restriction
being defined by the prior distribution. So, under the hypothesis that
the prior distribution is well-defined, the ``restrictive estimator''
is unbiased and has a smaller variance. On the other hand, for record
lengths greater than 15 years, $MLE, bHe$ and $bHo$ seems to be
equivalent.

\section{Conclusion}
\label{sec:ccl}

A framework to perform a regional Bayesian frequency analysis for
partially gauged stations is presented. The proposed model has the
advantage of being less restrictive than the most widely used regional
model, that is the Index Flood. Several case studies from French sites
were analysed to illustrate the superiority of the Bayesian approach
in comparison to the traditional Index Flood and to local approaches.
The influence of the homogeneousness level of the pooling group on
quantile estimates was also considered. Results demonstrate that
working with quite large and homogeneous regions rather than small and
strongly homogeneous regions is more efficient. Further work can focus
on the regional estimation of other characteristics of the flood
hydrograph. For instance, a regional Bayesian model can focus on Flood
Duration Frequency.

All statistical analysis was carried out in the \citet{Rsoft}
framework. For this purpose, two packages were contributed to this
software under the framework of the present research work. These two
packages integrate the tools that were developed to carry out the
modelling effort presented in this paper. The first one \textbf{POT}
performs statistical inference on Peaks Over Thresholds, while the
second one, \textbf{RFA}, contains several tools to carry out a
Regional Frequency Analysis. These two packages are available, free of
charge, at the web site \url{http://www.R-project.org}, section CRAN,
Packages.

\section*{Acknowledgements}
\label{sec:ackn}

The authors wish to thank the DIREN Rhône-Alpes for providing data. 
The authors are also very grateful to the two referees for
their constructive remarks which improve the document.

\appendix
\section*{Appendix A. Properties of the Index Flood on GP parameters} 
\label{sec:proof}

We provide in this appendix the proof for the following theorem:
\begin{thm}
  Let X be a random variable GP distributed. So $X$ has the
  Cumulative Distribution Function defined by:
  \begin{equation*}
  F(x) = 1 - \left( 1 +\frac{\xi\left(x - \mu\right)}{\sigma}\right)
  ^{-1/\xi}
\end{equation*}
Let $Y= C X$ where $C \in \mathbb{R}^+_*$. Then, $Y$ is also GP
distributed with parameters $\left(C\mu, C\sigma, \xi\right)$.
\end{thm}
\begin{proof}
  Let $X$ be a \textit{r.v.} GP distributed with parameters
  $\left(\mu, \sigma, \xi\right)$ and $Y= C X$ where $C \in
  \mathbb{R}^+_*$.  Then:
\begin{eqnarray*}
  \Pr\left[ Y \leq y \right] &=& \Pr\left[ X \leq \frac{y}{C} \right]\\
  &=& 1 - \left( 1 +
    \frac{\xi\left(\frac{y}{C} - \mu\right)}{\sigma} \right)^{-1/\xi}\\
  &=& 1 - \left( 1 +
    \frac{\xi\left(y - \mu C\right)}{\sigma C} \right)^{-1/\xi}
\end{eqnarray*}
So, $Y$ is also GP distributed with parameters $\left(\mu C, \sigma
  C, \xi \right)$.  The proof for the GEV case can be established in the
same way.
\end{proof}

\bibliography{biblio_ribatet}
\bibliographystyle{plainnat}
\end{document}